\documentclass[11pt,letterpaper]{article}

\usepackage[margin=1in]{geometry}
\usepackage{amsmath,enumerate}
\usepackage{amsfonts}
\usepackage{wrapfig}
\usepackage{float}
\usepackage{epsfig}
\usepackage{xspace}
\usepackage{paralist}
\usepackage{enumerate}
\usepackage{cases}
\usepackage{caption}
\usepackage{algorithmicx}
\usepackage[noend]{algpseudocode}
\usepackage[linesnumbered,lined,boxed,commentsnumbered]{algorithm2e}
\SetKwRepeat{Do}{do}{while}
\usepackage{color}
\usepackage{complexity}
\usepackage{multicol}
\newenvironment{proof}{\noindent {\bf Proof:  }}{\hfill\rule{2mm}{2mm}}

\newtheorem{claim}{Claim}[section]
\newtheorem{fact}{Fact}[section]
\newtheorem{theorem}{Theorem}[section]
\newtheorem{lemma}{Lemma}[section]
\newtheorem{corollary}{Corollary}[section]

\makeatletter 
\g@addto@macro{\@algocf@init}{\SetKwInOut{Parameter}{Parameters}} 
\makeatother

\newcommand*\samethanks[1][\value{footnote}]{\footnotemark[#1]}
\newcommand{\opt}{\mathsf{OPT}}
\newcommand{\alg}{\mathsf{ALG}}

\newcommand{\w}{\mathbf{w}}

\newcommand{\ceil}[1]{\ensuremath{\lceil #1 \rceil}}

\renewcommand{\v}{\mathsf{v}}

\newcommand{\one}{\mathbf{1}}
\newcommand{\vs}{\mathbf{s}}
\newcommand{\va}{\mathbf{a}}
\newcommand{\vx}{\mathbf{x}}
\newcommand{\ovpp}{\textsf{Online Vector Packing Problem}\xspace}
\newcommand{\ogap}{\textsf{Online Generalized Assignment Problem}\xspace}

\renewcommand{\R}{\mathbb{R}}

\title{Online Submodular Maximization Problem with Vector Packing Constraint}

\author{T-H. Hubert Chan\thanks{Department of Computer Science, the University of Hong Kong. {\texttt{\{hubert,sfjiang,zhtang,xwwu\}@cs.hku.hk}}} 
	\and Shaofeng H.-C. Jiang\samethanks
	\and Zhihao Gavin Tang\samethanks
	\and Xiaowei Wu\samethanks}
\date{}

\begin{document}
	
\begin{titlepage}
	\maketitle
	\begin{abstract}
We consider the online vector packing problem in which we have a $d$ dimensional knapsack and items $u$ with weight vectors $\w_u\in\R_+^d$ arrive online in an arbitrary order.
Upon the arrival of an item, the algorithm must decide immediately whether to discard or accept the item into the knapsack.
When item $u$ is accepted, $\w_u(i)$ units of capacity on dimension $i$ will be taken up, for each $i\in[d]$.
To satisfy the knapsack constraint, an accepted item can be later disposed of with no cost, but discarded or disposed of items cannot be recovered.
The objective is to maximize the utility of the accepted items $S$ at the end of the algorithm, which is given by $f(S)$ for some non-negative monotone submodular function $f$.

For any small constant $\epsilon > 0$, we consider the special case that the weight of an item on every dimension is at most a $(1-\epsilon)$ fraction of the total capacity, and give a polynomial-time deterministic $O(\frac{k}{\epsilon^2})$-competitive algorithm for the problem, where $k$ is the (column) sparsity of the weight vectors.  We also show several (almost) tight hardness results even when the algorithm is computationally unbounded.  We first show that under the $\epsilon$-slack assumption, no deterministic algorithm can obtain any $o(k)$ competitive ratio, and no randomized algorithm can obtain any $o(\frac{k}{\log k})$ competitive ratio. We then show that for the general case (when $\epsilon = 0$), no randomized algorithm can obtain any $o(k)$ competitive ratio.

In contrast to the $(1+\delta)$ competitive ratio achieved in Kesselheim et al. (STOC 2014) for the problem with random arrival order of items and under large capacity assumption, we show that in the arbitrary arrival order case, even when $\| \w_u \|_\infty$ is arbitrarily small for all items $u$, it is impossible to achieve any $o(\frac{\log k}{\log\log k})$ competitive ratio.
\end{abstract}
	\thispagestyle{empty}
\end{titlepage}

\section{Introduction}\label{sec:intro}

\paragraph{\ovpp.}
We consider the following online submodular maximization problem with vector packing constraint.
Suppose we have a $d$ dimensional knapsack, and items arrive online in an arbitrary order.
Each item $u\in\Omega$ has a \emph{weight vector} $\w_u \in \R_+^d$, i.e.,
when item $u\in\Omega$  is accepted,
for each $i \in [d]$, item~$u$ will take up $\w_u(i)$ units of capacity on every dimension~$i$ of the knapsack.
By rescaling the weight vectors, we can assume that each of the $d$ dimensions has capacity $1$.
Hence we can assume w.l.o.g. that $\w_u\in[0,1]^d$ for all $u\in\Omega$.
The \emph{(column) sparsity}~\cite{ipco/BansalKNS10,stoc/KesselheimTRV14} is defined as the minimum number~$k$ such that every weight vector $\w_u$ has at most $k$ non-zero coordinates.
The objective is to pack a subset of items with the maximum utility into the knapsack, where the utility of a set~$S$ of items is given by a non-negative monotone submodular function $f:2^\Omega\rightarrow \R_+$.

The vector packing constraint requires that the accepted items can take up a total amount of at most $1$ capacity
on each of the $d$ dimensions of the knapsack.
However, as items come in an arbitrary order, it can be easily shown that the competitive ratio is arbitrarily bad,
if the decision of acceptance of each item is decided online and cannot be revoked later.
In the literature, when the arrival order is arbitrary, the \emph{free disposal} feature~\cite{wine/FeldmanKMMP09}
is considered, namely, an accepted item can be disposed of when later items arrive.
On the other hand, we cannot recover items that are discarded or disposed of earlier.

We can also interpret the problem as solving the following program online,
where variables pertaining to $u$ arrive at step $u \in \Omega$.
We assume that the algorithm does not know the number of items in the sequence.
The variable $x_u\in\{0,1\}$ indicates whether item $u$ is accepted.
During the step $u$, the algorithm decides to set $x_u$ to $0$ or $1$,
and may decrease $x_{u'}$ from $1$ to $0$ for some $u' < u$ in order
to satisfy the vector packing constraints.
\begin{align*}
 \max \qquad f(\{u\in\Omega:x_u=1\}) & \\
 \text{s.t.}\qquad \sum_{u\in\Omega} \w_u(i) \cdot x_u \leq 1&, \qquad \forall i\in [d] \\
 x_u \in\{0,1\}&, \qquad \forall u \in \Omega.
\end{align*}

In some existing works~\cite{mor/BuchbinderN09,esa/FeldmanHKMS10,icalp/MolinaroR12,stoc/KesselheimTRV14},
the items are decided by the adversary, who sets the value (the utility of a set of items is the summation of their values) and the weight vector of each item, but the items arrive in a uniformly random order. This problem is sometimes referred to as \textsf{Online Packing LPs} with random arrival order, and each choice is irrevocable.
To emphasize our setting, we refer to our problem as \ovpp (with submodular objective and free disposal).

\paragraph{Competitive Ratio.}
After all items have arrived, suppose $S \subset\Omega$ is the set of items currently accepted (excluding those that are disposed of) by the algorithm. The objective is $\alg := f(S)$. Note that to guarantee feasibility, we have $\sum_{u\in S}\w_u\leq \one$, where $\one$ denotes the $d$ dimensional all-one vector.
The competitive ratio is defined as the ratio between the optimal objective $\opt$ that is achievable by an offline algorithm and the (expected) objective of the algorithm: $\mathsf{r} := \frac{\opt}{\E[\alg]}\geq 1$.

\subsection{Our Results and Techniques}

We first consider the \ovpp \emph{with slack}, i.e., there is a constant $\epsilon>0$ such that for all $u\in\Omega$, we have $\w_u\in[0,1-\epsilon]^d$, and propose a deterministic 
$O(\frac{k}{\epsilon^2})$-competitive algorithm, where $k$ is the sparsity of weight vectors.

\begin{theorem}\label{th:k_competitive_}
	For the \ovpp with $\epsilon$ slack,
	there is a (polynomial-time) deterministic $O(\frac{k}{\epsilon^2})$-competitive algorithm for the \ovpp.
\end{theorem}

Observe that by scaling weight vectors, Theorem~\ref{th:k_competitive_} implies a bi-criteria $(1+\epsilon, \frac{k}{\epsilon^2})$-competitive algorithm for general weight vectors, i.e., by relaxing the capacity constraint by an $\epsilon$ fraction, we can obtain a solution that is
$O(\frac{k}{\epsilon^2})$-competitive compared to the optimal solution (with the augmented capacity).

We show that our competitive ratio is optimal (up to a constant factor) for deterministic algorithms, and almost optimal (up to a logarithmic factor) for any (randomized) algorithms.
Moreover, all our following hardness results (Theorem~\ref{th:hardness_slack},~\ref{th:d-hardness} and~\ref{th:1-hardness}) hold for algorithms with \textbf{unbounded computational power}.

\begin{theorem}[Hardness with Slack]\label{th:hardness_slack}
	For the \ovpp with slack $\epsilon\in(0,\frac{1}{2})$, any deterministic algorithm has a competitive ratio $\Omega(k)$, even when the utility function is linear and all items have the same value, i.e., $f(S) := |S|$; for  randomized algorithms, the lower bound is $\Omega(\frac{k}{\log k})$.
\end{theorem}

We then consider the hardness of the \ovpp (without slack) and show that no (randomized) algorithm can achieve any $o(k)$-competitive ratio.
The following theorem is proved in Section~\ref{sec:hardness} by constructing a distribution of hard instances.
Since in our hard instances, the sparsity $k=d$, the following hardness result implies that it is impossible to achieve any $o(d)$ competitive ratio in general.

\begin{theorem}[Hardness without Slack]\label{th:d-hardness}
	Any (randomized) algorithm for the \ovpp has a competitive ratio $\Omega(k)$, even when $f(S) := |S|$.
\end{theorem}

As shown by~\cite{stoc/KesselheimTRV14}, for the \ovpp with random arrival order, if we have $\| \w_u \|_\infty = O(\frac{\epsilon^2}{\log k})$ for all items $u\in\Omega$, then a $(1+\epsilon)$ competitive ratio can be obtained.
Hence, a natural question is whether better ratio can be achieved under this ``small weight'' assumption.
For example, if $\max_{u\in\Omega}\{\|\w_u\|_\infty\}$ is arbitrarily small, is it possible to achieve a $(1+\epsilon)$ competitive ratio like existing works~\cite{esa/FeldmanHKMS10,ec/DevanurJSW11,icalp/MolinaroR12,stoc/KesselheimTRV14}?

Unfortunately, we show in Section~\ref{sec:hardness_light} that, even when all weights are arbitrarily small, it is still not possible to achieve any constant competitive ratio.

\begin{theorem}[Hardness under Small Weight Assumption]\label{th:1-hardness}
	There does not exist any (randomized) algorithm with an $o(\frac{\log k}{\log\log k})$ competitive ratio for the \ovpp, even when $\max_{u\in\Omega}\{\|\w_u\|_\infty\}$ is arbitrarily small and $f(S) := |S|$.
\end{theorem}

Our hardness result implies that even with free disposal, the problem with arbitrary arrival order is strictly harder than its counter part when the arrival order is random. 

\paragraph{Our Techniques.}
To handle submodular functions, we use the standard technique by considering marginal cost
of an item, thereby essentially reducing to linear objective functions.
However, observe that the hardness results in Theorems~\ref{th:hardness_slack},~\ref{th:d-hardness} and~\ref{th:1-hardness} hold even for linear objective function where every item has the same value.
The main difficulty of the problem comes from the weight vectors of items, i.e., when items conflict with one another due to multiple dimensions, it is difficult to decide which items to accept.  Indeed, even the offline version of the problem has an $\Omega(\frac{k}{\log k})$ \NP-hardness of approximation result~\cite{cc/HazanSS06,walcom/OualiFS11}.

For the case when $d=1$, i.e., $\w_u \in [0,1]$, it is very natural to compare items based on their \emph{densities}~\cite{tcs/HanKM15}, i.e., the value per unit of weight, and accept the items with maximum densities.  A naive solution is to use the maximum weight $\| w_u \|_\infty$
to reduce the problem to the $1$-dimensional case, but this can lead to $\Omega(d)$-competitive ratio, even though each weight vector has sparsity $k \ll d$.  To overcome this difficulty, we define for each item a density on \textbf{each} of the $d$ dimensions, and make items comparable on any particular dimension.

Even though our algorithm is deterministic, we borrow techniques from randomized algorithms
for a variant of the problem with matroid constraints~\cite{soda/BuchbinderFS15a,soda/ChanHJKT17}.  Our algorithm maintains a fractional solution, which is rounded at every step to achieve an integer solution.  When a new item arrives, we try to 
accept the item by \textbf{continuously} increasing its accepted fraction (up to $1$), while
for each of its $k$ non-zero dimensions, we decrease the fraction of
the currently least dense accepted item, as long as the rate of increase in value due to the new item is at least some factor times the rate of loss due to disposing of items fractionally.

The rounding is simple after every step.  If the new item is accepted with a fraction larger
than some threshold~$\alpha$, then the new item will be accepted completely in the integer solution; at the same time, if the fraction of some item drops below some threshold~$\beta$,
then the corresponding item will be disposed of completely in the integer solution.  The $\epsilon$ slack assumption is used to bound the loss of utility due to rounding.  The high level intuition of why the competitive ratio depends on the sparsity~$k$ (as opposed
to the total number~$d$ of dimensions) is that when a new item is fractionally increased,
at most~$k$ dimensions can cause other items to be fractionally disposed of.

Then, we apply a standard argument to compare the value of items that are eventually accepted (the utility of our algorithm) with the value of items that are \textbf{ever} accepted (but maybe disposed of later).  The value of the latter is in turn compared with that of an optimal solution to give the competitive ratio.

\subsection{Related Work}

The \ovpp (with free disposal) is general enough to subsume many well-known online problems.
For instance, the special case $d=1$ becomes the \textsf{Online Knapsack Problem}~\cite{tcs/HanKM15}.  The offline version of the problem captures the \textsf{$k$-Hypergraph $b$-Matching Problem} (with sparsity $k$ and
$\w_u\in\{0,\frac{1}{b}\}^d$, where $d$ is the number of vertices),
for which an $\Omega(\frac{k}{b\log k})$ \NP-hardness of approximation 
is known~\cite{cc/HazanSS06,walcom/OualiFS11}, for any $b\leq \frac{k}{\log k}$.
In contrast, our hardness results are due to the online nature of the problem and hold even
if the algorithms have unbounded computational power.

\paragraph{Free Disposal.}
The free disposal setting was first proposed by Feldman et al.~\cite{wine/FeldmanKMMP09} for the online edge-weighted bipartite matching problem with arbitrary arrival order, in which the decision whether an online node is matched to an offline node
must be made when the online node arrives.
However, an offline node can dispose of its currently matched node, if the new online node is more beneficial.
They showed that the competitive ratio approaches $1-\frac{1}{e}$ when the number of online nodes each offline node can accept approaches infinity.
It can be shown that in many online (edge-weighted) problems with arbitrary arrival order, no algorithm can achieve any bounded competitive ratio without the free disposal assumption.  Hence, this setting has been
adopted by many other works~\cite{soda/ConstantinFMP09,ec/BabaioffHK09,stacs/EpsteinLSW13,ec/Devanur0KMY13,algorithmica/HanKM14,scheduling/Fung14,esa/CharikarHN14,tcs/HanKM15,soda/BuchbinderFS15a}.

\paragraph{\ogap (OGAP).}
Feldman et al.~\cite{wine/FeldmanKMMP09} also considered a more general
online biparte matching problem,
where each edge $e$ has both a value $v_e$ and a weight $w_e$, and each offline node has a capacity constraint on the sum of weights of matched edges (assume without loss
of generality that all capacities are $1$).
It can be easily shown that the problem is a special case of the \ovpp with $d$ equal to the total number of nodes, and sparsity $k=2$:
every item represents an edge $e$, and has value $v_e$, weight $1$ on the dimension corresponding to the online endpoint, and weight $w_e$ on the dimension corresponding to the offline endpoint.

For the problem when each edge has arbitrary weight and each offline node has capacity 1,
it is well-known that the greedy algorithm that assigns each online node to the offline node with maximum marginal increase in the objective is $2$-competitive, while no algorithm is known to have a competitive ratio strictly smaller than $2$.
However, several special cases of the problem were analyzed and better competitive ratios have been achieved~\cite{soda/AggarwalGKM11,soda/DevanurJK13,esa/CharikarHN14,esa/AbolhassaniCCEH16}.

Apart from vector packing constraints, the online submodular maximization problem with
free disposal has been studied under matroid constraints~\cite{soda/BuchbinderFS15a,soda/ChanHJKT17}.
In particular, the uniform and the partition matroids can be thought of
special cases of vector packing constraints, where each item's weight vector has sparsity one and the same value for non-zero coordinate.  However, using special properties of partition matroids, the exact optimal competitive ratio can be derived in~\cite{soda/ChanHJKT17},
from which we also borrow relevant techniques to design our online algorithm.

\paragraph{Other Online Models.}
Kesselheim et al.~\cite{stoc/KesselheimTRV14} considered a variant of the problem when items (which they called requests) arrive in random order and have small weights compared to the total capacity;
this is also known as the \emph{secretary setting}, and free disposal is not allowed.
They considered a more general setting in which an item can be accepted with 
more than one option,
i.e., each item has different utilities and different weight vectors for different options.
For every $\delta\in (0,\frac{1}{2})$, for the case
when every weight vector is in $[0,\delta]^d$, they proposed an $O(k^{\frac{\delta}{1-\delta}})$-competitive algorithm, and a $(1+\epsilon)$-competitive algorithm when $\delta = O(\frac{\epsilon^2}{\log k})$, for $\epsilon\in(0,1)$.
In the random arrival order framework, many works assumed that the weights of items are much smaller than the total capacity~\cite{esa/FeldmanHKMS10,ec/DevanurJSW11,icalp/MolinaroR12,stoc/KesselheimTRV14}.  In comparison, our algorithm just needs the weaker $\epsilon$ slack assumption that no weight is more than $1 - \epsilon$ fraction of the total capacity.

\paragraph{Other Related Problems.}

The \textsf{Online Vector Bin Packing} problem~\cite{stoc/AzarCKS13,soda/AzarCFR16} is similar to the problem we consider in this paper.
In the problem, items (with weight $\w_u\in [0,1]^d$) arrive online in an arbitrary order and the objective is to pack all items into a minimum number of knapsacks, each with capacity $\one$.
The current best competitive ratio for the problem is $O(d)$~\cite{jct/GareyGJ76} while the best hardness result is $\Omega(d^{1-\epsilon})$~\cite{stoc/AzarCKS13}, for any constant $\epsilon>0$.

\paragraph{Future Work.}

We believe that it is an interesting open problem to see whether an $O(k)$-competitive ratio can be achieved for general instances, i.e., $\w_u\in[0,1]^d$. However, at least we know that it is impossible to do so using deterministic algorithms (see Lemma~\ref{lemma:hardness_determ}).

Actually, it is interesting to observe that
similar slack assumptions on the weight vectors of items have been made by several other literatures~\cite{siamcomp/ChekuriVZ14,stoc/AzarCKS13,stoc/KesselheimTRV14}.
For example, for the \textsf{Online Packing LPs} problem (with random arrival order)~\cite{stoc/KesselheimTRV14}, the competitive ratio $O(k^{\frac{\delta}{1-\delta}})$ holds only when $\w_u\in[0,\delta]^d$ for all $u\in\Omega$, for some $\delta\leq\frac{1}{2}$.
For the \textsf{Online Vector Bin Packing} problem~\cite{stoc/AzarCKS13}, while a hardness result $\Omega(d^{1-\epsilon})$ on the  competitive ratio is proof for general instances with $\w_u\in[0,1]^d$; when $\w_u \in [0,\frac{1}{B}]^d$ for some $B\geq 2$, they proposed an $O(d^{\frac{1}{B-1}}(\log d)^{\frac{B}{B-1}})$-competitive algorithm.

Another interesting open problem is whether the $O(k)$-competitive ratio can be improved for the problem under the ``small weight assumption''.
Note that we have shown in Theorem~\ref{th:1-hardness} that achieving a constant competitive ratio is impossible.

\section{Preliminaries}\label{sec:preli}

We use $\Omega$ to denote the set of items, which are not known by the algorithm
initially and arrive one by one.
Assume that each of the $d$ dimensions of the knapsack has capacity $1$.
For $u \in \Omega$, the weight vector $\w_u \in [0,1]^d$  is
known to the algorithm only when item $u$ arrives.
A set $S\subset\Omega$ of items  is \emph{feasible} if $\sum_{u\in S} \w_u \leq \one$.
The utility of $S$ is $f(S)$, where $f$ is a non-negative monotone submodular function.

For a positive integer $t$, we use $[t]$ to denote $\{1,2,\ldots,t\}$.

We say that an item $u$ is \emph{discarded} if it is not accepted when it arrives; it is \emph{disposed of} if it is accepted when it arrives, but later dropped to maintain feasibility.

Note that in general (without constant slack), no deterministic algorithm for the problem is competitive, even with linear utility function and when $d=k$.
A similar result when $k=1$ has been shown by Iwama and Zhang~\cite{approx/IwamaZ07}.

\begin{lemma}[Generalization of \cite{approx/IwamaZ07}]
\label{lemma:hardness_determ}
	Any deterministic algorithm has a competitive ratio $\Omega(\sqrt{\frac{k}{\epsilon}})$ for the \ovpp with weight vectors in $[0,1-\epsilon]^d$, even when the utility function is linear and $d=k$.
\end{lemma}
\begin{proof}
	Since the algorithm is deterministic, we can assume that the instance is adaptive.
	
	Consider the following instance with $k=d$.
	Let the first item have value $1$ and weight $1-\epsilon$ on all $d$ dimensions; the following (small) items have value $\sqrt{\frac{\epsilon}{k}}$ and weight $2\epsilon$ on one of the $d$ dimension (and $0$ otherwise).
	Stop the sequence immediately if the first item is not accepted.
	Otherwise let there be $\frac{1}{2\epsilon}$ items on each of the $d$ dimensions.
	Note that to accept any of the ``small'' items, the first item must be disposed of.
	We stop the sequence immediately once the first item is disposed of.
	
	It can be easily observe that we have either $\alg = 1$ and $\opt = \sqrt{\frac{k}{4\epsilon}}$, or $\alg = \sqrt{\frac{\epsilon}{k}}$ and $\opt\geq 1$, in both cases the competitive ratio is $\Omega(\sqrt{\frac{k}{\epsilon}})$.
\end{proof}

Note that the above hardness result (when $k=1$) also holds for the \ogap (with one offline node).

We use $\opt$ to denote both the optimal utility, and the feasible set that achieves this value.
The meaning will be clear from the context.

\section{Online Algorithm for Weight Vectors with Slack}\label{sec:k_sparse}

In this section, we give an online algorithm for weight vectors with constant slack~$\epsilon>0$.
Specifically, the algorithm is given some constant parameter $\epsilon > 0$ initially such that for all items $u \in \Omega$, its weight vector satisfies
$\|\w_u\|_\infty \leq	1 - \epsilon$.  On the other hand,
the algorithm does not need to know upfront the upper bound $k$
on the sparsity of the weight vectors.

\subsection{Deterministic Online Algorithm}

\paragraph{Notation.} 
During the execution of an algorithm, for each item $u \in \Omega$,
we use $S^u$ and $A^u$ to denote the feasible set of maintained items and the set of items that have ever been accepted, respectively, 
at the moment \textbf{just before} the arrival of item $u$.

We define the \emph{value} of $u$ as $\v(u) := f(u | A^u) = f(A^u\cup\{u\}) - f(A^u)$.
Note that the value of an item depends on the algorithm and the arrival order of items.
For $u\in \Omega$, for each $i \in [d]$, define the \emph{density} of $u$ at dimension $i$ as $\rho_u(i) := \frac{\v(u)}{\w_u(i)}$ if $\w_u(i) \neq 0$ and $\rho_u(i) := \infty$ otherwise.
By considering a lexicographical order on $\Omega$, 
we may assume that all ties in values and densities can be resolved consistently.

For a vector $\vx \in [0,1]^\Omega$, we use $\vx(u)$ to denote the component corresponding to coordinate $u\in \Omega$.
We overload the notation $\ceil{\vx}$ to mean
either the support $\ceil{\vx} := \{u \in \Omega : \vx(u) > 0\}$
or its indicator vector in $\{ 0,1 \}^\Omega$  such that $\ceil{\vx}(u) = \ceil{\vx(u)}$.

\paragraph{Online Algorithm.}
The details are given in Algorithm~\ref{alg:kcsparse}, which defines
the parameters $\beta := 1-\epsilon$, 
$\alpha := \sqrt{\beta} = 1 - \Theta(\epsilon)$ and $\gamma := \frac{1}{2}(1-\frac{\beta}{\alpha}) = \Theta(\epsilon)$.
The algorithm keeps a (fractional) vector $\vs \in [0,1]^\Omega$,
which is related to the actual feasible set $S$
maintained by the algorithm via the loop invariant
(of the \textbf{for} loop in lines~\ref{ln:for_begin}-\ref{ln:for_end}): 
$S = \lceil \vs \rceil$.  Specifically, when an item $u$ arrives,
the vector $\vs$ might be modified such that the coordinate $\vs(u)$ might
be increased and/or other coordinates might be decreased;
after one iteration of the loop, the feasible set $S$ is changed according
to the loop invariant.  The algorithm also maintains an auxiliary vector
$\va\in[0,1]^\Omega$ that keeps track of the maximum fraction of item $u$
that has ever been accepted.

\paragraph{Algorithm Intuition.}
The algorithm solves a fractional variant
behind the scenes using a linear objective function defined by $\v$.  For each dimension~$i \in [d]$,
it assumes that the capacity is $\beta < 1$.  Upon the arrival of a new element~$u \in \Omega$,
the algorithm tries to increase the fraction of item $u$ accepted via the parameter $\theta \in [0,1]$ in the 
\textbf{do}...\textbf{while} loop starting at line~\ref{ln:do_begin}.
For each dimension~$i \in [d]$ whose capacity is saturated (at $\beta$)
and $\w_u(i) > 0$, to further increase the fraction of item~$u$ accepted,
some item~$u^\theta_i$ with the least density $\rho_i$ will have its fraction decreased in order to make room for item $u$.
Hence, with respect to $\theta$, the value decreases at
a rate at most $\sum_{i} \w_u(i) \cdot \rho_i(u^\theta_i)$ due to disposing of
fractional items.
We keep on increasing $\theta$ as long as this rate of loss is
less than $\gamma$ times $\v(u)$ (which is the rate of increase in value due to item $u$). 

After trying to increase the fraction of item~$u$ (and disposing of other items fractionally),
the algorithm commits to this change only if at least $\alpha$ fraction
of item~$u$ is accepted, in which case any item whose accepted fraction is less than $\beta$ will be totally disposed of.

\begin{algorithm}[!ht]
	\SetAlgoLined
	\Parameter{$\alpha := \sqrt{1-\epsilon}, \beta := 1-\epsilon$, $\gamma := \frac{1}{2}(1-\sqrt{1-\epsilon})$}
	initialize $\vs, \va \in [0,1]^\Omega$ as all zero vectors; \Comment{$\lceil \vs \rceil$ is the current feasible solution}\\
	\For{each round when $u$ arrives \label{ln:for_begin}}{
		Define $\v(u) := f(u|\lceil \va \rceil)$\;
		Initialize $\theta \leftarrow 0$, $\vx^0 \leftarrow \vs$\;
		\Do{$\theta < 1$ and $\gamma\cdot \v(u) > \sum_{i\in [d]} \w_u(i) \cdot \rho_i(u_i^\theta)$ \label{ln:do_begin}}{
			Increase $\theta$ continuously (variables $\vx^\theta$ and $u^\theta_i$ all depend on $\theta$):
			
			\For{every $i\in[d]$}{
				\If{$\sum_{v\in \Omega}\vx^\theta(v)\w_v(i) = \beta$ and $\w_u(i) > 0$}{
					Set $u^\theta_i \leftarrow \arg\min \{ \rho_i(v): v\in \Omega\setminus\{u\}, \vx^\theta(v)\w_v(i) > 0 \}$\;
				}
				\If{$\sum_{v\in \Omega}\vx^\theta(v)\w_v(i) < \beta$ or $\w_u(i) = 0$}{
					Set $u^\theta_i \leftarrow \bot$ and $\rho_i(u_i^\theta) \leftarrow 0$\;
				}
			}
			Change $\vx^\theta(v)$ (for all $v\in \Omega$) at rate:
			$$
			\frac{d\vx^\theta(v)}{d\theta} = 
			\begin{cases}
			1, \quad v = u;\\
			- \max_{i\in[d]:u^\theta_i = v} \left\{ \frac{\w_u(i)}{\w_{u^\theta_i}(i)} \right\}, \quad v\in \{ u^\theta_i \}_{i\in[d]};\\
			0, \quad \text{otherwise.}
			\end{cases}
			$$\\
		}
		\If{$\theta \geq \alpha$}{
			$\vs \leftarrow \vx^\theta$, $\va(u) \leftarrow \vx^\theta(u)$;\Comment{update phase} \\
			\For{$v \in \Omega$ with $\vs(v) < \beta$}{
				$\vs(v) \leftarrow 0$;\label{line:remove_fraction} \Comment{dispose of small fractions} \\
			}
		}\Comment{if $\theta < \alpha$, then $\vs$ and $\va$ will not be changed}
	} \label{ln:for_end}
	\Return $\lceil \vs \rceil$.
	\caption{Online Algorithm}
	\label{alg:kcsparse}
\end{algorithm}

\subsection{Competitive Analysis} \label{sec:analysis}

For notational convenience, we use the superscripted versions (e.g.,
$\vs^u$, $\va^u$, $S^u = \ceil{\vs^u}$, $A^u = \ceil{\va^u}$) to indicate
the state of the variables at the beginning of the iteration in the
\textbf{for} loop (starting at line~\ref{ln:for_begin}) when item~$u$ arrives.
When we say the \textbf{for} loop, we mean the one that runs
from lines~\ref{ln:for_begin} to~\ref{ln:for_end}.
When the superscripts of the variables are removed (e.g., $S$ and $A$), we 
mean the variables at some moment just before or after an iteration of the \textbf{for} loop.

We first show that the following properties are loop invariants of the \textbf{for} loop.

\begin{lemma}[Feasibility Loop Invariant]
\label{lemma:feasibility}
The following properties are loop invariants of the \textbf{for}
loop:
\begin{compactitem}
\item[(a)] For every $i \in [d]$, $\sum_{v\in \Omega} \vs(v) \cdot \w_v(i) \leq \beta$, i.e., for every dimension, the total capacity consumed by the fractional solution $\vs$ is at most $\beta$.
\item[(b)] The set $S = \ceil{\vs} \subset \Omega$ is feasible for the original problem.
\end{compactitem}
\end{lemma}

\begin{proof}
Statement~(a) holds initially because $\vs$ is initialized to $\vec{0}$.  Next,
assume that for some item $u \in \Omega$, statement~(a) holds for $\vs^u$.
It suffices to analyze the non-trivial case when the changes to $\vs$ are committed
at the end of the iteration.  Hence, we show that statement~(a) holds
throughout the execution of the \textbf{do}...\textbf{while} loop 
starting at line~\ref{ln:do_begin}.  It is enough show that
for each $i \in [d]$, $g_i(\theta) := \sum_{v \in \Omega} \vx^\theta(v) \cdot \w_v(i) \leq \beta$ holds while $\theta$ is being increased.

To this end, it suffices to prove that
if $g_i(\theta) = \beta$, then $\frac{d g_i(\theta)}{d \theta} \leq 0$.
We only need to consider the case $\w_u(i) > 0$, because otherwise $g_i(\theta)$
cannot increase.
By the rules updating~$\vx$,
we have in this case
$\frac{d g_i(\theta)}{d \theta} \leq 
\frac{d \vx^\theta(u)}{d \theta} \w_u(i) + \frac{d \vx^\theta(u^\theta_i)}{d \theta} \w_{u^\theta_i}(i) 
\leq 0$, as required.

We next show that statement~(b) follows from statement~(a).
Line~\ref{line:remove_fraction} ensures that between iterations
of the \textbf{for} loop,
for all $v \in S = \ceil{\vs}$, $\vs(v) \geq \beta$.

Hence, for all $i \in [d]$,
we have $\sum_{v \in S} \w_v(i) \leq \frac{1}{\beta} \sum_{v \in S} \vs(v)\cdot \w_v(i)
= \frac{1}{\beta} \sum_{v \in \Omega} \vs(v)\cdot \w_v(i) \leq 1$,
where the last inequality follows from statement~(a).
\end{proof}

For a vector $\vx\in[0,1]^\Omega$, we define $\v(\vx) := \sum_{u\in \Omega} \v(u)\cdot \vx(u)$;
for a set $X\subset \Omega$, we define $\v(X) := \sum_{u\in X} \v(u)$.
Note that the definitions of $\v(\lceil \vx \rceil)$ are consistent under the set and the vector interpretations.

The following simple fact (which is similar to Lemma~2.1 of~\cite{soda/ChanHJKT17}) establishes the connection between the values of items (defined by our algorithm) and the utility of the solution (defined by the submodular function $f$).

\begin{fact}[Lemma 2.1 in \cite{soda/ChanHJKT17}]\label{fact:submo}
The \textbf{for} loop maintains
the invariants $f(A) = f(\emptyset) + \v(A)$ and $f(S) \geq f(\emptyset) + \v(S)$,
where $A = \ceil{\va}$ and $S = \ceil{\vs}$.
\end{fact}

Our analysis consists of two parts.
We first show that $\v(\va)$ is comparable to the value of our real solution $S$ in Lemma~\ref{lemma:S_to_A}.
Then, we compare in Lemma~\ref{lemma:OPT_to_A}  the value of an (offline) optimal solution with $\v(\va)$.
Combining the two lemmas we are able to prove Theorem~\ref{th:k_competitive}.

\begin{lemma}
\label{lemma:S_to_A}
The \textbf{for} loop maintains the invariant:
	$(1-\frac{\beta}{\alpha})\cdot \v(S) \geq (1-\frac{\beta}{\alpha}-\gamma)\cdot 
	\v(\va)$, where $S = \ceil{\vs}$.
	In particular, our choice of the parameters implies that
$\v(\va) \leq 2 \cdot \v(S)$.
\end{lemma}

\begin{proof}
We prove the stronger loop invariant that:
\begin{equation*}
\v(\vs) \geq (1 - \gamma - \frac{\beta}{\alpha}) \sum_{r \in A \setminus S} \v(r) \cdot \va(r)
+ (1 - \gamma) \sum_{r \in S} \v(r) \cdot \va(r),
\end{equation*}
where $S = \ceil{\vs}$ is the current feasible set and $A \setminus S$
is the set of items that have been accepted at some moment but are already discarded.

The invariant holds trivially initially when $S = A = \emptyset$
and $\vs = \vec{0}$.  Suppose the invariant
holds at the beginning of the iteration when item~$u \in \Omega$ arrives.
We analyze the non-trivial case when the item~$u$ is accepted into $S$,
i.e., $\vs$ and $\va$ are updated at the
end of the iteration.  Recall that $\vs^u$ and $\va^u$ refer to the 
variables at the beginning of the iteration,
and for the rest of the proof, we use the $\widehat{\vs}$ and $\widehat{\va}$ to denote
their states at the end of the iteration.

Suppose in the \textbf{do}...\textbf{while} loop,
the parameter $\theta$ is increased from $0$ to $\va(u) \geq \alpha$.
Since for all $r \neq u$, $\va^u(r) = \widehat{\va}(r)$,
we can denote this common value by $\va(r)$ without risk of ambiguity.
We use $\vx^u$ to denote the vector $\vx^\theta$ when $\theta = \va(u)$.
Then, we have
\begin{align*}
	\v(\vx^u) - \v(\vs^u)
	& \geq \v(u)\cdot \va(u) - \int_{0}^{\va(u)} \sum_{i\in[d]: u^\theta_i \neq \bot} (\frac{\w_u(i)}{\w_{u^\theta_i}(i)}\cdot \v(u^\theta_i)) d\theta \\
	& > \v(u)\cdot \va(u) - \int_{0}^{\va(u)} \gamma\cdot \v(u) d\theta \\
	& =  (1-\gamma)\cdot \v(u)\cdot \va(u),
\end{align*}
where the second inequality holds by the criteria of the \textbf{do}...\textbf{while} loop.

Next, we consider the change in value $\v(\widehat{\vs}) - \v(\vx^u)$,
because some (fractional) items are disposed of in line~\ref{line:remove_fraction}.
Let $D \subseteq S^u$ be such discarded items.
Since an item is discarded only if its fraction is less than $\beta$,
the value lost is at most $\beta \sum_{r \in D} \v(r) \leq \frac{\beta}{\alpha} \sum_{r \in D} \v(r) \cdot \va(r)$, where the last inequality follows
because $\va(r) \geq \alpha$ for all items~$r$ that are ever accepted.
Therefore, we have
$$\v(\widehat{\vs}) - \v(\vx^u)	 \geq  -\frac{\beta}{\alpha} \sum_{r \in D} \v(r) \cdot \va(r).
$$

Combining the above two inequalities, we have
	$$\v(\widehat{\vs}) - \v(\vs^u)
	\geq (1-\gamma)\cdot \v(u)\cdot \va(u) -\frac{\beta}{\alpha} \sum_{r \in D} \v(r) \cdot \va(r).$$
	
Hence, using the induction hypothesis that the loop invariant holds
at the beginning of the iteration,
it follows that 

\begin{align*}
\v(\widehat{\vs}) \geq & 
(1 - \gamma - \frac{\beta}{\alpha}) \sum_{r \in A^u \setminus S^u} \v(r) \cdot \va(r) 
+ (1 - \gamma) \sum_{r \in S^u} \v(r) \cdot \va(r)+ (1-\gamma)\cdot \v(u)\cdot \va(u) \\
& -\frac{\beta}{\alpha} \sum_{r \in D} \v(r) \cdot \va(r) \\
\geq & (1 - \gamma - \frac{\beta}{\alpha}) \sum_{r \in \widehat{A} \setminus \widehat{S}} \v(r) \cdot \va(r)
+ (1 - \gamma) \sum_{r \in \widehat{S}} \v(r) \cdot \va(r),
\end{align*}
where $\widehat{A} = \ceil{\widehat{\va}}$
and $\widehat{S} = \ceil{\widehat{\vs}}$, as required.

We next show that the stronger invariant implies the result of the lemma.
Rewriting the invariant gives
\begin{equation*}
\v(\vs) \geq (1 - \gamma - \frac{\beta}{\alpha}) \sum_{r \in A } \v(r) \cdot \va(r)
+ \frac{\beta}{\alpha} \sum_{r \in S} \v(r) \cdot \va(r)
\geq (1 - \gamma - \frac{\beta}{\alpha}) \sum_{r \in A } \v(r) \cdot \va(r)
+ \frac{\beta}{\alpha} \cdot \v(\vs),
\end{equation*}
where the last inequality follows because $\va(r) \geq \vs(r)$ for all $r \in S$.
Finally, the lemma follows because $\v(S) = \v(\ceil{\vs}) \geq \v(\vs)$.
\end{proof}

The following lemma gives an upper bound on the value
of the items in a feasible set that are discarded right away by the algorithm.

\begin{lemma}
\label{lemma:OPT_to_A}
The \textbf{for} loop maintains the invariant that
if $\opt$ is a feasible subset of items that have arrived so far,
then
	$\gamma \cdot \v(\opt\setminus A) \leq \frac{k}{\beta(1-\alpha)}\cdot \v(\va)$,
	where $A = \ceil{\va}$.
In particular, our choice of the parameters implies that
$\v(\opt\setminus A) \leq O(\frac{k}{\epsilon^2}) \cdot \v(S)$.
\end{lemma}
\begin{proof}
	Consider some $u\in \opt\setminus A$. Since $u\notin A$, 
	in iteration~$u$ of the \textbf{for} loop,
	we know that at the end of the \textbf{do}...\textbf{while} loop,
	 we must have $\theta < \alpha$, which implies $\gamma \cdot \v(u) \leq \sum_{i\in [d]} \w_u(i)\cdot \rho_i(u_i^\theta)$ at this moment.
	
	Recall that by definition, 
	$\rho_i(u^\theta_i)$ is either (i)~0 in the case $\sum_{v \in \Omega} \vx^\theta(v) \cdot \w_v(i) < \beta$ and $\w_u(i) > 0$, or (ii)~the minimum density $\rho_i(v)$ in dimension~$i$ among items $v \neq u$ such that
	$\vx^\theta(v)\cdot\w_v(i) > 0$.
	
	Hence, in the second case, we have
	\begin{align*}
	\rho_i(u^\theta_i) & \leq 
	\frac{\sum_{v \neq u: \vx^\theta(v)\w_v(i) > 0} \vx^\theta(v) \cdot \v(v)  }{\sum_{v \neq u: \vx^\theta(v)\w_v(i) > 0} \vx^\theta(v) \cdot \w_v(i)} =
	\frac{\sum_{v \neq u: \vx^\theta(v)\w_v(i) > 0} \vx^\theta(v) \cdot \v(v)  }{\beta - \theta \cdot \w_u(i)} \\
	& \leq \frac{\sum_{v: \w_v(i) > 0} \va(v) \cdot \v(v)}{\beta (1 - \alpha)} = \frac{V_i}{\beta (1 - \alpha)},
	\end{align*}
	where $V_i := \sum_{v: \w_v(i) > 0} \va(v) \cdot \v(v)$ depends only on the current $\va$ and $i \in [d]$.
	In the last inequality, we use $\theta \cdot \w_u(i) \leq \alpha \beta$ and a
	very loose upper bound on the numerator.  Observe that
	for the case (i) $\rho_i(u^\theta_i) = 0$, the inequality $\rho_i(u^\theta_i) \leq \frac{V_i}{\beta (1 - \alpha)}$ holds trivially.

Hence, using this uniform upper bound on $\rho_i(u^\theta_i)$,
we have $\gamma \cdot \v(u) \leq \sum_{i \in [d]} \w_u(i) \cdot \frac{V_i}{\beta (1 - \alpha)}$.
	
Therefore, we have
\begin{align*}
\gamma \cdot \v(\opt \setminus A) & \leq \sum_{u \in \opt \setminus A} \sum_{i \in [d]} \w_u(i) \cdot \frac{V_i}{\beta (1 - \alpha)}
=  \sum_{i \in [d]}  \left( \sum_{u \in \opt \setminus A} \w_u(i) \right) \cdot \frac{V_i}{\beta (1 - \alpha)} \\
& \leq \sum_{i \in [d]} \frac{V_i}{\beta (1 - \alpha)} \leq \frac{k \cdot \v(\va)}{\beta(1 - \alpha),}
\end{align*}
where the second to last inequality follows because $\opt \setminus A$ is feasible,
and $\sum_{i \in [d]} V_i \leq k \cdot \v(\va)$, because
for each $v \in \Omega$, $|\{i \in [d]: \w_v(i) > 0 \}| \leq k$.
\end{proof}

\begin{theorem}\label{th:k_competitive}
Algorithm~\ref{alg:kcsparse} is $O(\frac{k}{\epsilon^2})$-competitive.
\end{theorem}
\begin{proof}
Suppose $\opt$ is a feasible subset.  Recall that $S$ is the feasible subset currently maintained by the algorithm.
Then, by the monotonicity and the submodularity of $f$, we have
$f(\opt) \leq	f(\opt \cup A) \leq f(A) + \sum_{u \in \opt\setminus A} f(u|A) \leq  f(\emptyset) + \v(A) + \v(\opt\setminus A)$,
where we use Fact~\ref{fact:submo}
and submodularity $f(u|A) \leq f(u|A^u) = \v(u)$
in the last inequality.

Next, observe that for all $u \in A$, $\va(u) \geq \alpha$.
Hence, we have $\v(A) \leq \frac{\v(\va)}{\alpha} = O(1) \cdot \v(\va)$.
Combining with Lemma~\ref{lemma:OPT_to_A},
we have $f(\opt) \leq f(\emptyset) + O(\frac{k}{\epsilon^2}) \cdot \v(\va)$.

Finally, using Lemma~\ref{lemma:S_to_A} and Fact~\ref{fact:submo}
gives $f(\opt) \leq O(\frac{k}{\epsilon^2}) \cdot f(S)$, as required.
\end{proof}

\subsection{Hardness Results: Proof of Theorem~\ref{th:hardness_slack}}

We show that for the \ovpp with slack $\epsilon\in (0,\frac{1}{2})$,
no deterministic algorithm can achieve $o(k)$-competitive ratio, and
no randomized algorithm can achieve $o(\frac{k}{\log k})$-competitive ratio.
To prove the hardness result for randomized algorithms, we apply Yao's principle~\cite{focs/Yao77} and
construct a distribution of hard instances,
such that any deterministic algorithm cannot perform well in expectation.
Specifically, we shall show that each instance in the support of the distribution
has offline optimal value $\Theta(\frac{k}{\log k})$, but any deterministic algorithm
has expected objective value $O(1)$, thereby proving Theorem~\ref{th:hardness_slack}.

In our hard instances, the utility function is linear, and all items have the same value, i.e., the utility function is $f(S) := |S|$.
Moreover, we assume all weight vectors are in $\{0,1-\epsilon\}^d$, for any arbitrary $\epsilon\in(0,\frac{1}{2})$.
Hence, we only need to describe the arrival order of items, and the non-zero dimensions of weight vectors.  In particular, we can associate each item~$u$
with a $k$-subset of $[d]$.  We use ${[d] \choose k}$ to denote the collection
of $k$-subsets of $[d]$.

\paragraph{Notations.}
We say that two items are \emph{conflicting},
if they both have non-zero weights on some dimension~$i$
(in which case, we say that they \emph{conflict} with each other on dimension $i$).
We call two items \emph{non-conflicting} if they do not conflict with each other on any dimension.

Our hard instances show that in some case when items conflict with one another on different dimensions, the algorithm might be forced to make difficult decisions on choosing which item to accept.
By utilizing the nature of unknown future, we show that it is very unlikely for any algorithm to make the right decisions on the hard instances.
Although accepted  items can be later disposed of to make room for (better) items, by carefully setting the weights and arrival order, we show that disposing of accepted  items cannot help to get a better objective (hence in a sense, disabling free-disposal). 

\paragraph{Hard instance for deterministic algorithms.}
Let $d := 2 k^2$.
Recall that each item is specified by an element of  ${[d] \choose k}$, indicating
which $k$ dimensions are non-zero.
Consider any deterministic algorithm. An arriving sequence of length at most $2k$ is chosen adaptively.
The first item is picked arbitrarily, and the algorithm must select this item, or else the sequence stops immediately.  Subsequently, in each round, the non-zero dimensions for the next arriving item~$u$ are picked according to the following rules.
\begin{compactitem}
\item  Exactly $k-1$ dimensions from $[d]$ are chosen such that no previous item has
picked them.
\item Suppose $\widehat{u} \in {[d] \choose k}$ is the item currently kept by the algorithm.  Then, the remaining dimension~$i$ is picked from $\widehat{u}$ such that no other arrived item
conflicts with $\widehat{u}$ on dimension~$i$.  If no such dimension~$i$ can be picked,
then the sequence stops.
\end{compactitem}

\begin{lemma}
Any deterministic algorithm can keep at most 1 item, while
there exist at least $k$ items that are mutually non-conflicting, implying
that an offline optimal solution contains at least $k$ items.
\end{lemma}

\begin{proof} By adversarial choice, every arriving item conflicts
with the item currently kept by the algorithm.  Hence, the algorithm
can keep at most 1 item at any time.

We next show that when the sequence stops, there exist at least
$k$~items in the sequence that are mutually non-conflicting.  For the case when there are $2k$ items in the sequence, consider the items in reversed order of arrival.  Observe that each item conflicts with only one item that arrives before it. Hence, we can scan the items one by one backwards, and
while processing a remaining item, we remove any earlier item that conflicts with it.  After we finish with the scan, there are at least $k$ items remaining that are mutually non-conflicting.

Suppose the sequence stops with less than $2k$ items.  It must be the case
that while we are trying to add a new item~$u$, we cannot
find a dimension~$i$ contained in the item~$\widehat{u}$ currently kept by the algorithm such that no already arrived item conflicts with~$\widehat{u}$
on dimension~$i$.  This implies that for every non-zero dimension~$i$ of~$\widehat{u}$, there is already an item $u_i$ conflicting with~$\widehat{u}$
on that dimension.  Since by choice, each dimension can cause a conflict
between at most 2 items, these $k$ items $u_i$'s must be mutually non-conflicting.
\end{proof}

\paragraph{Distribution of Hard Instances.}
To use Yao's principle~\cite{focs/Yao77}, we give a procedure to
sample a random sequence of items.  For some large enough integer~$\ell$
that is a power of~$2$,
define $k:= 100 \ell \log_2 \ell + 1$, which is the sparsity of the weight vectors.
Observe that $\ell = \Theta(\frac{k}{\log k})$,
and define $d := \ell + 400 \ell^2 \log_2 \ell = O(\frac{k^2}{\log k})$
to be the number of dimensions. 
We express the set of dimensions $[d] = I \cup J$
as the disjoint union of $I := [\ell]$ and $J := [d] \setminus I$.
The items arrive in~$\ell$ phases, and for each~$i \in [\ell]$,
$4\ell - i + 1$~items arrive. Recall that each item is characterized by its ~$k$ non-zero dimensions (where the non-zero coordinates
all equal $1 - \epsilon > \frac{1}{2}$).  We initialize $J_1 := J$.  For $i$ from $1$ to $\ell$, we describe how the items in phase~$i$ are sampled as follows.
\begin{compactitem}
\item[1.] Each of the~$4\ell-i+1$ items will have $i \in I = [\ell]$ as the only non-zero dimension in $I$.
\item[2.] Observe that (inductively) we have $|J_i| = (4\ell - i + 1) \cdot 100 \ell \log_2 \ell$.  We partition $J_i$ randomly into $4\ell-i+1$ disjoint subsets,
each of size exactly $k-1= 100 \ell \log_2 \ell$.  Each such subset corresponds
to the remaining $(k-1)$ non-zero dimensions of an item in phase~$i$.
These items in phase~$i$ can be revealed to the algorithm one by one.

\item[3.] Pick $S_i$ from those $4\ell - i+ 1$ subsets uniformly at random;
define $J_{i+1} := J_i \setminus S_i$.  Observe that the algorithm
does not know $S_i$ until the next phase~$i+1$ begins.
\end{compactitem}

\begin{claim} \label{claim:opt_k}
In the above procedure, the items corresponding to $S_i$'s for $i \in [\ell]$
are mutually non-conflicting.  This implies that there is an offline
optimal solution containing $\ell = \Theta(\frac{k}{\log k})$ items.
We say that those $\ell$ items are \emph{good}, while other items are \emph{bad}.
\end{claim}

We next show that bad items are very likely to be conflicting.

\begin{lemma}
\label{lemma:bad_conflict}
Let $\mathcal{E}$ be the event that there exist
two bad items that are non-conflicting.
Then, $\Pr[\mathcal{E}] \leq \frac{1}{\ell^2}$.
\end{lemma}

\begin{proof}
An alternative view of the sampling process is that
the subsets $S_1, S_2, \ldots, S_\ell$ are first sampled for the
good items.  Then, the remaining bad items
can be sampled independently across different phases (but
note that items within the same phase are sampled in a dependent way).

Suppose we condition on the subsets $S_1, S_2, \ldots, S_\ell$ already sampled.
Consider phases $i$ and $j$, where $i < j$.  Next, we further condition on
all the random subsets generated in phase~$j$ for defining the corresponding items.  We fix some bad item~$v$ in phase $j$.

We next use the remaining randomness (for picking
the items) in phase~$i$.  Recall that each bad item in phase~$i$ corresponds to a random subset of size $k-1= 100 \ell \log_2 \ell$
in $J_i \setminus S_i$, where $|J_i \setminus S_i| \leq 4(k-1)\ell$.
If we focus on such a particular (random) subset from phase~$i$,
the probability that it is disjoint from the subset corresponding to item~$v$ (that we fixed from phase~$j$) is at most
$(1 - \frac{k-1}{4(k-1)\ell})^{k-1} \leq \exp(-25 \log_2 \ell) \leq \frac{1}{\ell^7}$.

Observe that there are in total at most $4\ell^2$ items.
Hence, taking a union over all possible pairs of bad items,
the probability of the event~$\mathcal{E}$
is at most $(4\ell^2)^2 \cdot \frac{1}{\ell^7} \leq \frac{1}{\ell^2}$.
\end{proof}

\begin{lemma}
\label{lemma:exp_alg}
For any deterministic algorithm~$\alg$ applied to the above random
procedure, the expected number of items kept in the end is $O(1)$.
\end{lemma}

\begin{proof}
Let $X$ denote the number of good items and $Y$ denote the number of bad items kept by the algorithm at the end.

Observe that the sampling procedure allows the good item (corresponding to $S_i$) in phase~$i$ to be decided after the deterministic algorithm
finishes making all its decisions in phase~$i$.  Hence, the probability
that the algorithm keeps the good item corresponding to $S_i$
is at most $\frac{1}{4\ell - i + 1} \leq \frac{1}{3\ell}$.  Since this holds for every phase, it follows that $E[X] \leq \frac{1}{3\ell} \cdot \ell = \frac{1}{3}$.

Observe that conditioning on the complementing event~$\overline{\mathcal{E}}$ (refer to
Lemma~\ref{lemma:bad_conflict}), at most 1 bad item can be kept by the algorithm, because any two bad items are conflicting.  Finally, because the total number of items is at most $4\ell^2$,
we have $E[Y] = \Pr[\mathcal{E}] E[Y| \mathcal{E}] +
\Pr[\overline{\mathcal{E}}] E[Y| \overline{\mathcal{E}}]
\leq \frac{1}{\ell^2} \cdot 4 \ell^2 + 1 \cdot 1 \leq 5$.

Hence, $E[X+Y] \leq 6$, as required.
\end{proof}

\begin{corollary}
By Claim~\ref{claim:opt_k} and Lemma~\ref{lemma:exp_alg},
Yao's principle implies that for any randomized algorithm,
there exists a sequence of items such that the value
of an offline optimum is at least $\Theta(\frac{k}{\log k})$,
but the expected value achieved by the algorithm is $O(1)$.
\end{corollary}

\section{Hardness Result: Proof of Theorem~\ref{th:d-hardness}}\label{sec:hardness}

In this section we prove Theorem~\ref{th:d-hardness}, a slightly better hardness result for the \ovpp with no slack.
We show that for the problem with weight vectors in $[0,1]^d$, no randomized algorithm can achieve any $o(k)$-competitive ratio.
By Yao's principle~\cite{focs/Yao77},
we construct a distribution of hard instances (with sparsity $k=d$),
such that any deterministic algorithm cannot perform well in expectation.
Specifically, we shall show that each instance in the support of the distribution
has offline optimal value $\Theta(d)$, but any deterministic algorithm
has expected objective value $O(1)$, 
thereby proving Theorem~\ref{th:d-hardness}.

\paragraph{Distribution of Hard Instances.}
In our instances, we set $f(S) := |S|$, i.e., the objective function is linear, and each item has the same value $1$.
The hardness of the problem is due to the dimension of the vector packing constraints.  We describe how to sample a random instance.  Let $\delta := 2^{-2d}$, and suppose
$\sigma = (\sigma_1, \ldots, \sigma_d)$ is a permutation on $[d]$ sampled uniformly at random.
Items arrive in $d$ phases, and in phase $t \in [d]$, there are $d+1-t$ items.
Each item in phase $t$ has a weight vector in $[0,1]^d$ of the following form:
\begin{compactitem}
\item[1.] The coordinates in $\{\sigma_1, \ldots, \sigma_{t-1}\}$
have value 0.
\item[2.] There are $d+1-t$ coordinates with positive values.  Exactly one has value
$1-(2^t-1)\delta$, while the remaining $d-t$ coordinates have
value $2^t\delta$.
\end{compactitem}
Observe that in phase $t$, there are exactly $d+1-t$ weight vectors of this form.
After seeing the first item in phase $t$,
the algorithm knows $\sigma_1, \sigma_2, \ldots \sigma_{t-1}$,
and can deduce what the $d+1-t$ items in this phase are like.  
Hence, we may assume that the items in phase $t$ arrive according to the
lexicographical order on the weight vectors.

To describe the structure of the weight vectors, in phase $t$, for $j \in [d+1-t]$, we use $(t,j)$ to index an item.
It is evident that it is impossible to fit $2$ weight vectors from the same phase simultaneously.  Therefore,
we can also interpret that in phase $t$, there are $d+1-t$ options and the algorithm
can choose at most 1 of them.

\begin{figure}[htb]
		\vspace*{-20pt}
		\centering
		\includegraphics*[width = \textwidth]{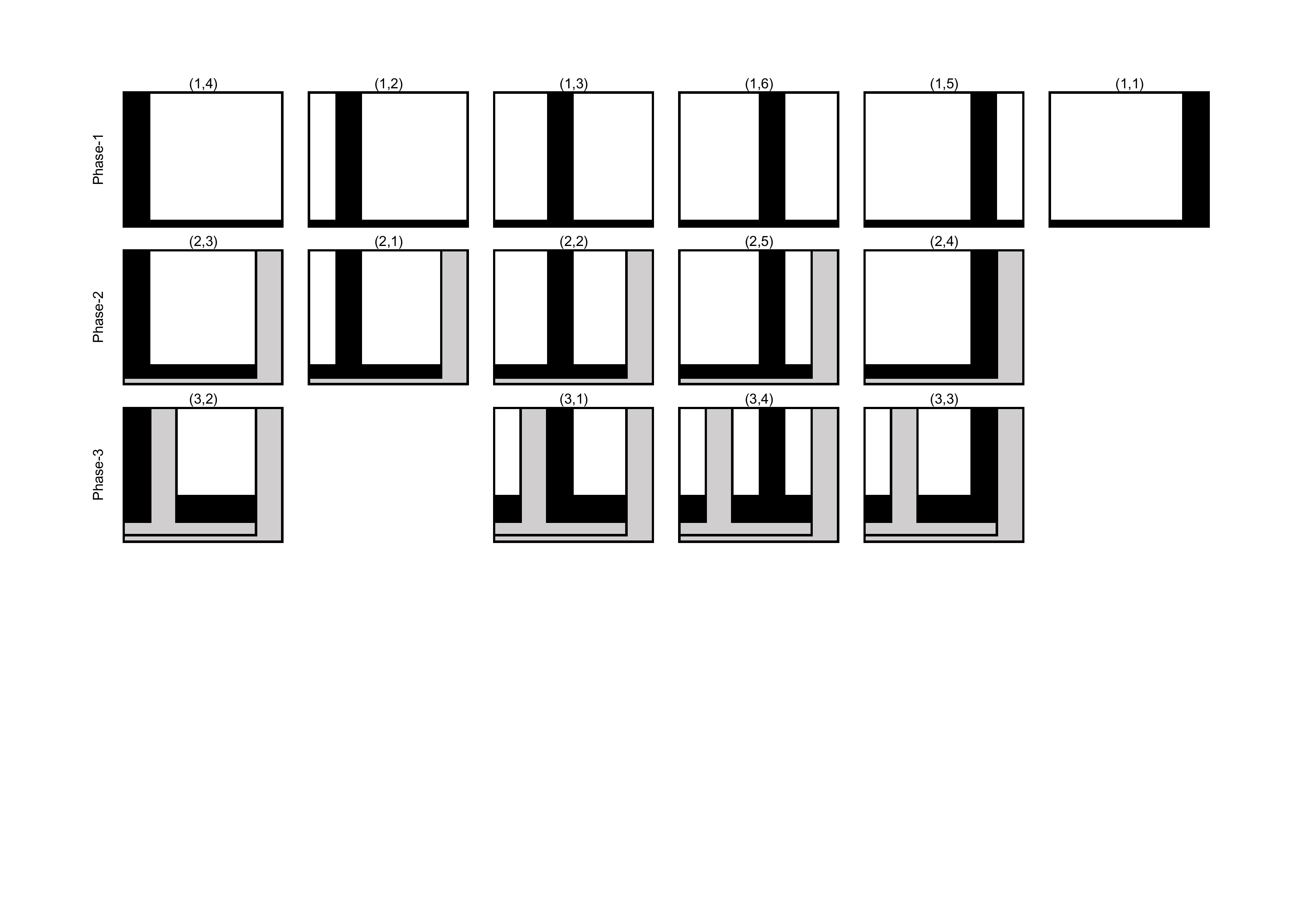}
		\vspace*{-4.5cm}
		\caption{Hard instance: in the illustrative example, we have $d=6$ and $\sigma = (6,2,3,1,5,4)$. The weight of each item is represented by the black histogram (the height of each bar represents the weight at each dimension), while the gray histograms represent the item currently accepted by the optimal solution.}\label{fig:hard-instance-1}
\end{figure}
	
However, observe that the index $j$ does not
indicate the arrival order within a phase, otherwise information about the permutation $\sigma$ would be leaked.
The weight vector of item $(t,j)$ on each dimension $i \in [d]$ is defined as follows (refer to Figure~\ref{fig:hard-instance-1}).
\begin{equation*}
	\w_{t,j}(i) = \begin{cases}
	0, &i\in\{ \sigma_1,\ldots,\sigma_{t-1} \}\\
	1-(2^t-1)\delta, &i=\sigma_{t+j-1}\\
	2^t\delta, &\text{otherwise}.
	\end{cases}
\end{equation*}
	
\paragraph{Optimal solution.}
In the optimal solution, we can accept items $\{(t,1)\}_{t \in [d]}$ simultaneously with utility $d$.
To check feasibility, for each $i \in [d]$, in dimension $\sigma_i$, the total weight is:
\begin{equation*}
	\sum_{t\in[d]} \w_{t,1}(\sigma_i) = \sum_{t<i} 2^t\delta + (1-(2^i-1)\delta) = 1-\delta<1.
\end{equation*}
	
\paragraph{Alternative Setup.}
To facilitate the analysis, we explain why the following assumptions can be justified.
\begin{compactitem}
	\item[1.] As explained before, at the beginning
	of phase $t \in [d]$, the algorithm already knows $\sigma_1, \ldots, \sigma_{t-1}$,
	and also the $d+1-t$ weight vectors that arrive in phase $t$.
	\item[2.] In phase $t$, we can assume that the algorithm 
	decides whether to choose one out of the $d+1-t$ items.
	\item[3.] After the algorithm makes the decision in phase $t$,
	then $\sigma_t$ is drawn uniformly at random from $[d] \setminus \{\sigma_1, \ldots, \sigma_{t-1}\}$.  Observe this randomness is oblivious to the choice
	of the algorithm during phase $t$, and hence is equivalent to the original description.
	\item[4.] After $\sigma_t$ is revealed to the algorithm (at the end of phase~$t$), then the algorithm knows whether it has chosen $(t,1)$ in phase~$t$.
	Observe in the last phase~$d$, there is only one item~$(d,1)$.
	We shall argue using the next lemma to see that the algorithm gains something from phase~$t$ only if it has chosen $(t,1)$.
\end{compactitem}
	
\begin{lemma}[Conflicting Items]\label{lemma:conflict}
		Suppose in phase~$t < d$, the algorithm has chosen an item that is not~$(t,1)$.
		Then, unless the algorithm disposes of this chosen item, it cannot take any item in phase~$t+1$.
		
		Moreover, if the algorithm has chosen a item in phase~$t$ that is not $(t,1)$, then it can replace this item
		with another item in phase~$t+1$ without losing anything.
\end{lemma}
\begin{proof}
	If the algorithm has chosen an item $u$ in phase~$t$ that is not $(t,1)$,
	then there is some dimension~$i$, whose consumption is
	at least~$1-(2^t-1)\delta$.  Observe that in phase~$t+1$,
	every weight vector will have a value at least $2^{t+1} \delta$ in this dimension~$i$.
	Hence, unless the chosen item $u$ in phase~$t$ is disposed of,
	no item from phase~$t+1$ can be chosen.
	
	Observe that in phase~$t+1$, there is an item $\widehat{u}$ whose weight vector can be obtained by transforming from $\w_u$ as follows.
	\begin{compactitem}
	\item[1.] The coordinate corresponding to $\sigma_t$ is set to 0.
	\item[2.] The coordinate with value $1 - (2^t-1) \delta$ is decreased to
	$1 - (2^{t+1} - 1) \delta$.
	\item[3.] Every other coordinate with value $2^t \delta$ is increased
	to $2^{t+1} \delta$.
	\end{compactitem}
	
	If we replace item $u$ from phase~$t$ with
	item $\widehat{u}$ from phase~$t+1$, certainly the first two items of the weight vector modifications will not hurt.  Observe that the third item will not have any adverse effect in the future, because for phase $\widehat{t} > t+1$,
	the largest value in the coordinate is at most $1 - (2^{\widehat{t}} - 1)\delta$.
	Hence, item~3 will not create additional conflicts in the future.
\end{proof}
	
\paragraph{Warm Up Analysis.}
In our alternative setup, because $\sigma_t$ is revealed only at the end of phase~$t$,
it follows that in phase~$t$, the probability that the algorithm chooses $(t,1)$ is $\frac{1}{d+1-t}$.
Hence, the expected value gain by the algorithm is $\sum_{t=1}^d \frac{1}{d+1-t} = \Theta(\log d)$.

\paragraph{Improved Analysis.} 
Observe that the probability of accepting the ``correct'' item in each phase increases as $t$ increases.	Suppose
we stop after phase $\frac{d}{2}$.  Then, the optimal value is $\frac{d}{2}$.
However, using Lemma~\ref{lemma:conflict}, the expected value gain is
at most $1 + \sum_{t=1}^{\frac{d}{2} -1 } \frac{1}{d+1-t} = \Theta(1)$,
where we can assume that the algorithm can gain value 1 in phase~$\frac{d}{2}$. 
This completes the proof of Theorem~\ref{th:d-hardness} (observe that the sparsity of weight vectors $k=d$).

\section{Hardness under Small Weight Assumption}\label{sec:hardness_light}

As we have shown in Section~\ref{sec:k_sparse}, as long as $\max_{u\in\Omega}\{ \|\w_u\|_\infty\}$ is strictly smaller than $1$ (by some constant),
we can obtain an $O(k)$ competitive ratio for the \ovpp (with arbitrary arrival order and free disposal).
However, for random arrival order, it is shown~\cite{stoc/KesselheimTRV14} that
the competitive ratio can be arbitrarily close to $1$ if each coordinate of the weight vector
is small enough.

We describe in this section a distribution of hard instances (with sparsity $k=\Theta(d)$ and $f(S) := |S|$) such that 
any deterministic algorithm has an expected objective value $O(\frac{\log\log d}{\log d}\cdot \opt)$,
where $\opt$ is the offline optimal value, hence (by Yao's principle) proving Theorem~\ref{th:1-hardness}.

\paragraph{Distribution of Hard Instances.}
Each item $u$ in our hard instances has unit utility and weight vector $\w_u \in \{0,\epsilon\}^d$, where $\epsilon>0$ can be any arbitrarily small number (not necessarily constant).
Hence, an item is characterized by its non-zero dimensions (of its weight vector).
The sparsity $k$ of the weight vectors is at most $d$.
We fix some integer $\ell$, and let items arrive in $\ell$ phases.
Define the total number of dimensions $d := \frac{(2l)!}{l!} + l$, i.e., $l = \Theta(\frac{\log d}{\log \log d})$.  We write collection of dimensions $[d] := I \cup J$ as the disjoint
union of $I := [\ell]$ and $J := [d] \setminus I$.
For each phase~$i \in [\ell]$, we define $b_i := \frac{(2 \ell - i +1)!}{\ell !}$.
We initialize $J_1 := J$.
We describe how the items are generated in each phase~$i$, starting from $i = 1$ to $\ell$:

\begin{enumerate}
\item[1.] Each item in phase~$i$ has non-zero dimension~$i \in I$,
while every dimension in $I \setminus \{i\}$ is zero.

\item[2.] For the rest of the dimensions in $J$, only dimensions in $J_i$ can be non-zero.
The set $J_i$ has size $b_i$, and is partitioned in an arbitrary, possibly deterministic, manner into $(2 \ell - i + 1)$ subsets $\{S^{(i)}_j: j \in [2 \ell -i +1]\}$, each of size $b_{i+1}$.  Each such subset $S^{(i)}_j$ defines $\frac{1}{\epsilon}$ copies
of a type of item, which has non-zero dimensions in $J_i \setminus S^{(i)}_j$ (together
with dimension~$i$ in the previous step).  Hence, in total, phase~$i$
has $\frac{1}{\epsilon} \times (2 \ell  - i +1)$ items that can arrive in an arbitrary order.

\item[3.] If $i < \ell$, at the end of phase~$i$, we pick $\sigma_i \in [2 \ell - i +1]$
uniformly at random, and set $J_{i+1} := S^{(i)}_{\sigma_i}$.
\end{enumerate}

\begin{claim}[Offline Optimal Solution]
\label{claim:opt_eps}
Given any sequence of items, one can include the all $\frac{1}{\epsilon}$ items
of the type corresponding to $S^{(i)}_{\sigma_i}$ from each phase~$i \in [\ell]$
to form a feasible set of items.
Hence, there is an offline optimal solution has at least $\opt \geq \frac{\ell}{\epsilon} = \Theta(\frac{\log d}{\epsilon \log \log d})$ items.
\end{claim}

\begin{lemma}
\label{lemma:alg_eps}
For any deterministic algorithm~$\alg$ applied to the above random
procedure, the expected number of items kept in the end is
at most $O(\frac{1}{\epsilon}) \leq O(\frac{\log \log d}{\log d})$.
\end{lemma}

\begin{proof}
By construction, at most $\frac{1}{\epsilon}$ items can be accepted in each phase~$i$,
because of dimension~$i \in I$.
Fix some $1 \leq i < \ell$. Observe that $\sigma_i$ can be chosen at the end of phase~$i$,
after the algorithm has made all its decisions for phase~$i$.
An important observation is that any item chosen in phase~$i$
that has a non-zero dimension in $S^{(i)}_{\sigma_i}$ will
conflict with all items arriving in later phases.  Because there
are at least $\frac{1}{\epsilon}$ potential items that we can take in the last phase,
we can assume that in phase~$i$, any item accepted that has a non-zero dimension
in $S^{(i)}_{\sigma_i}$ is disposed of later.

Hence, for $1 \leq i < \ell$, each item accepted in phase~$i$ by the algorithm will
remain till the end with probability  $\frac{1}{2 \ell - i + 1}$.

Therefore, the expected number of items that are kept by the algorithm at the end is at most
$\frac{1}{\epsilon} (1 + \sum_{i = 1}^{\ell - 1} \frac{1}{2 \ell - i + 1}) = O(\frac{1}{\epsilon})$, as required.
\end{proof}

\begin{corollary}
By Claim~\ref{claim:opt_eps} and Lemma~\ref{lemma:alg_eps},
Yao's principle implies that for any randomized algorithm,
there exists a sequence of items such that
the expected value achieved by the algorithm is
at most $O(\frac{\log \log d}{\log d}) \cdot \opt$,
where $\opt$ is the value achieved by an offline optimal solution.
\end{corollary}

{
\bibliography{packing}
\bibliographystyle{alpha}
}

\end{document}